\documentclass[runningheads]{llncs}
\usepackage{graphicx}
\usepackage[T1]{fontenc}
\usepackage{lmodern}
\usepackage[utf8]{inputenc}
\usepackage{csquotes}
\usepackage{microtype}
\usepackage{mathtools}
\usepackage{amsmath}
\usepackage{amssymb}
\usepackage{booktabs}
\usepackage{multirow}
\usepackage{wrapfig}
\usepackage[margin,noinline,draft]{fixme}
\usepackage{cryptocode}
\usepackage{tikz}
\usetikzlibrary{matrix}
\usetikzlibrary{graphs}
\usetikzlibrary{positioning, arrows.meta, calc}

\usepackage{xspace}
\newtheorem{conj}{Conjecture}
\newcommand{\ie}{i.\,e.\xspace}

\newcommand{\ILP}{\textsc{Integer Linear Programming}}
\newcommand{\nfold}{$n$-fold}

\newcommand{\TwoStage}{$2$-stage stochastic}

\newcommand{\QC}{\textsc{Quadratic Congruences}}
\newcommand{\Rest}{\textsc{Non-Unique Remainder}}
\newcommand{\enc}{\text{enc}}
\newcommand{\QCInd}{\text{QC}}
\newcommand{\MRD}{\text{NR}}
\newcommand{\SatInd}{\text{$3$}}
\newcommand{\TwoStageProb}{\textsc{$2$-stage ILP}}

\AtEndEnvironment{proof}{\phantom{}\qed}

\begin{document}
\title{The Double Exponential Runtime is Tight for 2-Stage Stochastic ILPs\thanks{This work was supported by the DFG projects JA 612/20-1 and KL 3408/1-1}}

\authorrunning{K. Jansen et al.}

\author{Klaus Jansen\inst{1} \and
Kim-Manuel Klein\inst{2} \and
Alexandra Lassota \inst{3}}

\institute{Department of Computer Science, Kiel University,  Kiel, Germany
\email{kj@informatik.uni-kiel.de}\\ \and
Department of Computer Science, Kiel University,  Kiel, Germany\\
\email{kmk@informatik.uni-kiel.de}\\ \and
Department of Computer Science, Kiel University,  Kiel, Germany\\
\email{ala@informatik.uni-kiel.de}}

\maketitle
\begin{abstract}
We consider fundamental algorithmic number theoretic problems and their relation to a class of block structured Integer Linear Programs (ILPs) called $2$-stage stochastic. A $2$-stage stochastic ILP is an integer program of the form $\min \{c^T x \mid \mathcal{A} x = b, \ell \leq x \leq u, x \in \mathbb{Z}^{r + ns} \}$ where the constraint matrix $\mathcal{A} \in \mathbb{Z}^{nt \times r +ns}$ consists of $n$ matrices $A_i \in \mathbb{Z}^{t \times r}$ on the vertical line and $n$ matrices $B_i \in \mathbb{Z}^{t \times s}$ on the diagonal line aside. 

First, we show a stronger hardness result for a number theoretic problem called \QC{} where the objective is to compute a number $z \leq \gamma$ satisfying $z^2 \equiv \alpha \bmod \beta$ for given $\alpha, \beta, \gamma \in \mathbb{Z}$. This problem was proven to be NP-hard already in 1978 by Manders and Adleman. However, this hardness only applies for instances where the prime factorization of $\beta$ admits large multiplicities of each prime number. We circumvent this necessity proving that the problem remains NP-hard, even if each prime number only occurs constantly often.

Then, using this new hardness result for the $\QC$ problem, we prove a lower bound of $2^{2^{\delta(s+t)}} |I|^{O(1)}$ for some $\delta > 0$ for the running time of any algorithm solving $2$-stage stochastic ILPs assuming the  Exponential Time Hypothesis (ETH). Here, $|I|$ is the encoding length of the instance. This result even holds if $r$, $||b||_{\infty}$, $||c||_{\infty}, ||\ell||_{\infty}$ and the largest absolute value $\Delta$ in the constraint matrix~$\mathcal{A}$ are constant. This shows that the state-of-the-art algorithms are nearly tight. Further, it proves the suspicion that these ILPs are indeed harder to solve than the closely related \nfold{} ILPs where the contraint matrix is the transpose of~$\mathcal A$.

\keywords{$2$-Stage Stochastic ILPs, Quadratic Congruences, Lower Bound, Exponential Time Hypothesis}
\end{abstract}

\section{Introduction}
One of the most fundamental problems in algorithm theory and optimization is the \ILP{} problem. Many theoretical and practical problems can be modeled as integer linear programs (ILPs) and thus they serve as a very general but powerful framework for tackling various questions. Formally, the \ILP{} problem is defined as
\begin{align*}
    \min\{c^\top x \,|\, \mathcal{A}x = b, \ell \leq x \leq u, x \in \mathbb{Z}^{d_2}\}
\end{align*}
for some matrix $\mathcal{A} \in \mathbb{Z}^{d_1 \times d_2}$, a right-hand side $b \in \mathbb{Z}^{d_1}$, an objective function~$c \in \mathbb{Z}^{d_2}$ and some lower and upper bounds $\ell, u \in \mathbb{Z}^{d_2}$. The goal is to find a solution~$x$ such that the value of the objective function $c^\top x$ is minimized. In general, this problem is NP-hard. Thus, it is of great interest to find structures to these ILPs which make them solvable more efficiently. In this work, we consider \TwoStage{} integer linear programs where the constraint matrix admits a specific block structure. Namely, the constraint matrix $\mathcal{A}$ only contains non-zero entries in the first few columns and block-wise along the the diagonal aside. This yields the following form:
\begin{equation*}
\mathcal A =
\begin{pmatrix}
A_1	    & B_1   & 0     & \dots	    & 0      \\
A_2	    & 0 	& B_2   & \ddots    & \vdots \\
\vdots  & \vdots & \ddots & \ddots    & 0      \\
A_n 	& 0     & \dots & 0     	& B_n
\end{pmatrix}.
\end{equation*}
Thereby $A_1, \dots, A_n \in \mathbb{Z}^{t \times r}$ and $B_1, \dots, B_n \in \mathbb{Z}^{t \times s}$ are integer matrices themselves. The complete constraint matrix $\mathcal{A}$ has size $nt \times r +ns$. Let $\Delta$ denote the largest absolute entry in $\mathcal{A}$.

Such \TwoStage{} ILPs are a common tool in stochastic programming and they are often used in practice to model uncertainty of decision making over time~\cite{DBLP:journals/eor/Albareda-SambolaVF06,DBLP:journals/mor/DempsterFJLLK83,kall1994stochastic,DBLP:journals/ior/LaporteLM94}. Due to the applicability a lot of research has been done in order to solve these (mixed) ILPs efficiently in practice. Since we focus on the theoretical aspects of \TwoStage{} ILPs in this chapter, we only refer the reader to the surveys \cite{GAVENCIAK2020100596,kuccukyavuz2017introduction,twostage_survey} and the references therein regarding the practical methods.

The current state-of-the-art algorithms to solve \TwoStage{} ILPs admits a running time of $2^{(2\Delta)^{r^2s+rs^2}} n\log^3(n) \cdot |I|$ where $|I|$ is the binary encoding length of the input~\cite{DBLP:journals/corr/abs-1904-01361} or respectively of $n \log^{O(rs)}(n) 2^{(2\Delta)^{O(r^2 + rs)}}$~\cite{DBLP:journals/corr/abs-2012-11742} by a recent result. The first result improves upon the result in~\cite{DBLP:conf/ipco/Klein20} due to Klein where the dependence on $n$ was quadratic. The dependencies on the block dimensions and $|I|$ were similar. 
The first result in that respect was by Hemmecke and Schulz~\cite{Hemmecke_two_stage_03} who provided an algorithm with a running time of $f(r,s,t, \Delta) \cdot \text{poly}(n)$ for some computable function~$f$. However, due to the use of an existential result from commutative algebra, no explicit bound could be stated for $f$. 

Let us turn our attention to the \nfold{} ILPs for a moment, which where first introduced in \cite{DBLP:journals/disopt/LoeraHOW08}. These ILPs admit a constraint matrix which is the transpose of the \TwoStage{} constraint matrix. Despite being so closely related, \nfold{} ILPs can be solved in time near linear in the number of blocks and only single exponentially in the block-dimensions of $A_i^T, B_i^T$ \cite{DBLP:journals/corr/abs-2002-07745,jansen2018near_linear}. 

Thus, it is an intrinsic questions whether we can solve \TwoStage{} ILPs more efficient or -- as the latest algorithms suggest -- whether \TwoStage{} ILPs are indeed harder to solve than the closely related \nfold{} ILPs. We answer this question by showing a double-exponential lower bound in the running time for any algorithm solving the \textsc{\TwoStage{} integer linear programming} (\TwoStageProb) problem. Here, the \TwoStageProb{} problem is the corresponding decision variant which asks whether the ILP admits a feasible solution.  

To prove this hardness, we reduce from the \QC{} problem. This problem asks whether there exists a $z \leq \gamma$ such that $z^2 \equiv \alpha \bmod \beta$ for some $\gamma, \alpha, \beta \in \mathbb{N}$. This problem was proven to be NP-hard  by Manders and Adleman~\cite{DBLP:journals/jcss/MandersA78} already in 1978 by showing a reduction from $3$-\textsc{SAT}. This hardness even persists if the prime factorization of $\beta$ is given~\cite{DBLP:journals/jcss/MandersA78}. By this result, Manders and Adleman prove that it is NP-complete to compute the solutions of diophantine equations of degree $2$. However, their reduction yields large parameters. In detail, the occurrences of each prime factor in the prime factorization of $\beta$ is too large to obtain the desired lower bound for the \TwoStageProb{} problem. The occurrence of each prime factor is at least linear in the number of variables and clauses of the underlying $3$-\textsc{SAT} problem.

We give a new reduction yielding a stronger statement: The \QC{} problem is NP-hard even if the prime factorization of $\beta$ is given and each prime factor occurs at most once (except 2 which occurs four times). Beside being useful to prove the lower bounds for solving the \TwoStage{} ILPs, we think this results is of independent interest. We obtain a neat structure which may be helpful in various related problems or may yield stronger statements of past results which use the \QC{} problem.

In order to achieve the desired lower bounds on the running time we make use of the Exponential Time Hypothesis (ETH) -- a widely believed conjecture stating that the $3$-\textsc{SAT} problem cannot be solved in subexponentially time with respect to the number of variables:

\begin{conj}[ETH \cite{DBLP:journals/jcss/ImpagliazzoP01}]
The $3$-\textsc{SAT} problem cannot be solved in time less than $O(2^{\delta_\SatInd n_\SatInd})$ for some constant $\delta_\SatInd > 0$ where $n_\SatInd$ is the number of variables in the instance.
\end{conj}

Note that we use the index $3$ for all variables of the $3$-\textsc{SAT} problem.

Using the ETH, plenty lower bounds for various problems are shown, for an overview on the techniques and results see e.g.~\cite{DBLP:books/sp/CyganFKLMPPS15}. 
So far, the best algorithm runs in time $O(2^{0.387 n_{\SatInd}})$, \ie, it follows that $\delta_\SatInd \leq 0.387$ \cite{DBLP:books/sp/CyganFKLMPPS15}.

In the following, we also need the Chinese Remainder Theorem (CRT) for some of the proofs, which states the following:

\begin{proposition}[CRT \cite{DBLP:books/daglib/0068082}]
Let $n_1, \dots, n_k$ be pairwise co-prime. Further, let $i_1, \dots, i_k$ be some integers. Then there exists integers $x$ satisfying $x \equiv i_j \bmod n_j$ for all $j$. Further, any two solutions $x_1$, $x_2$ are congruent modulo $\prod_{j=1}^k n_j$. 
\end{proposition}

\paragraph*{Summary of Results}
\begin{itemize}
\item We give a new reduction from the $3$-\textsc{SAT} problem to the \QC{} problem which proves a stronger NP-hardness result: The \QC{} problem remains NP-hard, even if the prime factorization of $\beta$ is given and each prime number greater than $2$ occurs at most once and the prime number $2$ occurs four times. This does not follow from the original proof. In contrast, the original proof generates each prime factor at least $O(n_\SatInd + m_\SatInd)$ times, where $m_\SatInd$ is the number of clauses in the formula. Our reduction circumvents this necessity, yet neither introduces noteworthily more nor larger prime factors. The proof is based on the original one. We believe this result is of independent interest.

\item Based on this new reduction, we show strong NP-hardness for the so-called \Rest{} problem. In this algorithmic number theoretic problem we are given $x_1, \dots, x_{n_{\MRD}}, y_1, \dots, y_{n_{\MRD}}, \zeta \in \mathbb{N}$ and pairwise coprime numbers $q_1, \dots, q_{n_{\MRD}}$. The question is to decide whether there exists a number $z \in \mathbb{Z}_{>0}$ with $z \leq \zeta$ satisfying the following congruences:
\begin{align*}
& z \bmod q_1 \in \{x_1, y_1\}  \\    
& z \bmod q_2 \in \{x_2, y_2\} \\  
& \vdots \\
& z \bmod q_{n_{\MRD}} \in \{x_{n_{\MRD}}, y_{n_{\MRD}}\} .
\end{align*}
In other words, either the residue $x_i$ or $y_i$ should be met for each equation. 
This problem is a natural generalization of the Chinese Remainder problem where $x_i = y_i$ for all $i$. In that case, however, the problem can be solved using the Extended Euclidean algorithm. To the best of our knowledge the \Rest{} problem has not been considered in the literature so far.  

\item Finally, we show that the \Rest{} problem can be modeled by a \TwoStage{} ILP. Assuming the ETH, we can then conclude a doubly exponential lower bound of $2^{2^{\delta(s+t)}}|I|^{O(1)}$ on the running time for any algorithm solving \TwoStage{} ILPs. The double exponential lower bound even holds if  $r=1$ and $\Delta, ||b||_{\infty}, ||c||_{\infty}  \in O(1)$. This proves the suspicion that \TwoStage{} ILPs are significantly harder to solve than \nfold{} ILPs with respect to the dimensions of the block matrices and $\Delta$. Furthermore, it implies that the current state-of-the-art algorithms for solving \TwoStage{} ILPs is indeed (nearly) optimal.
\end{itemize}

\paragraph*{Further Related Work}
In recent years there was significant progress in the development of algorithms for $n$-fold ILPs and lower bounds on the other hand. Assume the parameters as of the transpose of the \TwoStage{} constraint matrix, \ie, the blocks $A_i^T$ in the first few rows have dimension $r \times t$ and the blocks $B_i^T$ along the diagonal beneath admit a dimension of $s \times t$. The best known algorithms to solve these ILPs have a running time of $2^{O(rs^2)}(rs\Delta)^{O(r^2s+s^2)} (nt)^{1+o(1)}$~\cite{DBLP:journals/corr/abs-2002-07745} or respectively a running time of $(rs\Delta)^{r^2s+s^2} L^2 (nt)^{1+o(1)} $~\cite{jansen2018near_linear} where $L$ denotes the encoding length of the largest number in the input. The best known lower bound is $\Delta^{\delta_{\text{\nfold{}}}(r+s)^2}$ for some $\delta_{\text{\nfold{}}} > 0$~\cite{DBLP:journals/corr/abs-1904-01361}.

Despite their similarity, it seems that \TwoStage{} ILPs are significantly harder to solve than \nfold{} ILPs. Yet, no superexponential lower bound for the running time of any algorithm solving the \TwoStageProb{} problem was shown. There is a lower bound for a more general class of ILPs in~\cite{DBLP:journals/corr/abs-1904-01361} that contain \TwoStage{} ILPs showing that the running time is double-exponential parameterized by the topological height of the treedepth decomposition of the primal or dual graph. However, the topological height of \TwoStage{} ILPs is constant and thus no strong lower bound can be derived for this case.

If we relax the necessity of an integral solution, the \TwoStage{} LP problem becomes solvable in time $2^{2\Delta^{O(t^3)}} n \log^3(n) \log(||u-\ell||_{\infty})$ $\log(||c||_{\infty})$~\cite{DBLP:journals/corr/abs-1912-03501}. For the case of mixed integer linear programs there exists an algorithm solving \TwoStage{} MILPs in time $2^{\Delta^{\Delta^{t^{O(t^2)}}}} n \log^3(n) \log(||u-\ell||_{\infty}) \log(||c||_{\infty})$ \cite{DBLP:journals/corr/abs-1912-03501}. Both results rely on the fractionality of a solution, whose size is only dependent on the parameters. This allows us to scale the problem such that it becomes an ILP (as the solution has to be integral) and thus state-of-the-art algorithms for \TwoStage{} ILPs can be applied.

There are also studies for a more general case called 4-Block ILPs where the constraint matrix consists of non-zero entries in the first few columns, the first few rows and block-wise along the diagonal. This may be seen as the combination of \nfold{} and \TwoStage{} ILPs. Only little is known about them: They are in XP \cite{DBLP:conf/ipco/HemmeckeKW10}. Further, a lower and upper bound on the Graver Basis elements (inclusion-wise minimal kernel elements) of $O(n^{r} f(k,\Delta))$ was shown recently \cite{DBLP:conf/esa/0009K0S20}, where $r$ is the number of rows in the submatrix appearing repeatedly in the first few rows and $k$ denotes the sum of the remaining block dimensions.

\paragraph*{Structure of this Chapter}
Section~\ref{sec:2Stage:QCHardness} presents the stronger hardness result for the \QC{} problem we derive by giving a new reduction from the $3$-\textsc{SAT} problem. Then we show that the \QC{} problem can be modeled as a \TwoStage{} ILP in Section~\ref{sec:RedTo2Stage}. To do so, we introduce a new problem called the \Rest{} problem as an intermediate step during the reduction. Finally, in Section~\ref{sec:ETH} we bring the reductions together to prove the desired lower bound. This involves a construction which lowers the absolute value of $\Delta$ at the cost of slightly larger block dimensions.

\section{Advanced Hardness for \QC{}}\label{sec:2Stage:QCHardness}
This section proves that every instance of the $3$-\textsc{SAT} problem can be transformed into an equivalent instance of the \QC{} problem in polynomial time. Recall that the \QC{} problem asks whether there exists a number $z \leq \gamma$ such that $z^2 \equiv \alpha \bmod \beta$ holds. This problem was proven to be NP-hard by Manders and Adleman~\cite{DBLP:journals/jcss/MandersA78} showing a reduction from $3$-SAT. This hardness even persists when the prime factorization of $\beta$ is given~\cite{DBLP:journals/jcss/MandersA78}. However, we aim for an even stronger statement: The \QC{} problem remains NP-hard even if the prime factorization of $\beta$ is given and each prime number greater than $2$ occurs at most once and the prime number $2$ occurs four times. This does not follow from the original hardness proof. In contrast, if $n_\SatInd$ is the number of variables and $m_\SatInd$ the number of clauses in the $3$-SAT formula then $\beta$ admits a prime factorization with $O(n_\SatInd+m_\SatInd)$ different prime numbers each with a multiplicity of at least $O(n_\SatInd+m_\SatInd)$.
Even though our new reduction lowers the occurrence of each prime factor greatly, we neither introduces noteworthily more nor larger prime factors.

While the structure of our proof resembles that of the original one from~\cite{DBLP:journals/jcss/MandersA78}, adapting it to our needs requires various new observations concerning the behaviour of the newly generated prime factors and the functions we introduce. The original proof heavily depends on the numbers being high powers of the prime factors whereas we employ careful combinations of (new) prime factors. This requires us to introduce other number theoretical results into the arguments.

In the following, before presenting the reduction and showing its correctness formally, we want to give an idea of the hardness proof. The reduction may seem non-intuitive at first as it only shows the final result of equivalent transformations between various problems until we reach the \QC{} one. In the following, we list all these problems in order of their appearance whose strong NP-hardness is shown implicitly along the way. Afterwards, we give short ideas of their respective equivalence, which is then proved formally in separate claims in the next theorem. Note that not all variables are declared at this point, but also not necessary to understand the proof sketch.

\begin{itemize}
\item ($3$-\textsc{SAT}) Is there a truth assignment $\eta$ that satisfies all clauses $\sigma_k$ of the $3$-SAT formula $\Phi$ simultaneously?
\item (P2) Are there values $y_k \in \{0,1,2,3\}$ and a truth assignment $\eta$ such that $0 = y_k - \sum_{x_i \in \sigma_k} \eta(x_i) - \sum_{\bar{x_i} \in \sigma_k} (1 - \eta(x_i)) +1$ for all $k$?
\item (P3) Are there values $\alpha_j \in \{-1, +1\}$ such that $\sum_{j=0}^\nu \theta_j \alpha_j \equiv \tau \bmod 2^3 \cdot p^* \prod_{i=1}^{m'} p_i$ for some $\theta_j$ and $\tau$ specified in dependence on the formula later on and some prime numbers $p_i$ and $p^*$?
\item (P5) Is there an $x \in \mathbb{Z}$ satisfying
\begin{align*}
\tag{P5.1}
& 0 \leq |x| \leq H \\
\tag{P5.2}
& x \equiv \tau \bmod 2^3 \cdot p^* \prod_{i=1}^{m'} p_i \\
\tag{P5.3}
& (H+x)(H-x) \equiv 0 \bmod K
\end{align*}
for some $H$ dependent on the $\theta_j$ and $K$ being a product of primes?
\item (P6)  Is there an $x \in \mathbb{Z}$ satisfying
\begin{align*}
\tag{P6.1}
& 0 \leq |x| \leq H\\
\tag{P6.2}
& (\tau-x)(\tau+x) \equiv 0 \bmod 2^4 \cdot p^* \prod_{i=1}^{m'} p_i \\
\tag{P6.3}
& (H+x)(H-x) \equiv 0 \bmod K ?
\end{align*}
\item (\QC{}) Is there a number $x \leq H$ such that $(2^4 \cdot p^* \cdot  \prod_{i=1}^{m'} p_i + K)x^2 \equiv K\tau^2 + 2^4 \cdot p^* \cdot  \prod_{i=1}^{m'} p_i H^2 \bmod 2^4 \cdot p^* \cdot  \prod_{i=1}^{m'} p_i \cdot K ?$
\end{itemize}

The $3$-\textsc{SAT} problem is transformed to Problem (P2) by using the straight-forward interpretation of truth values as numbers $0$ and $1$ and the satisfiability of a clause as the sum of its literals being larger zero. Introducing slack variables~$y_k$ yields the above form.

Multiplying each equation of (P2) with exponentially growing factors and then forming their sum preserves the equivalence of these systems. Introducing some modulo consisting of unique prime factors larger than the outcome of the largest possible sum obviously does not influence the system. Replacing the variables $\eta(x_i)$ and $y_k$ by variables $\alpha_j$ with domain $\{-1, +1\}$, re-arranging the term and defining parts of the formula as the variables $\theta_j$ and $\tau$ yields Problem~(P3).

We then introduce some Problem (P4)  to integrate the condition $x \leq H$. The problem asks whether there exists some $x \in \mathbb{Z}$ such that 
\begin{align*}
\tag{P4.1}
& 0 \leq |x| \leq H \\
\tag{P4.2}
& (H + x) (H -x) \equiv 0 \bmod K?
\end{align*}
By showing that each solution to the system (P4) is of form $\sum_{j=0}^\nu \theta_j \alpha_j$ we can combine (P3) and (P4) yielding (P5).

Using some observations about the form of solutions for the second constraint of Problem (P5) we can re-formulate it as Problem (P6).

Next, we use the fact that $p^* \prod_{i=1}^{m'} p_i$ and $K$ are co-prime per definition and thus we can combine (P6.2) and (P6.3) to one equivalent equation. To do so, we take each left-hand side of (P6.2) and (P6.3) and multiply the modulo of the respective other equation and form their overall sum. Using a little re-arranging this finally yields the desired \QC{} problem. 

Before we finally present the reduction, we first prove a lemma about the size of the product of prime numbers, which comes in handy in the respective theorem. 

\begin{lemma}\label{l:ProdPrimes}
Denote by $q_i$ the $i$th prime number. The product of the first $k$ prime numbers $\prod_{i=1}^k q_i$ is bounded by $2^{2k \log(k)}$ for all $k \geq 2$.
\end{lemma}
\begin{proof}
Denote by $\pi(x)$ the number of prime numbers of size at most $x$. It holds that $\pi(x) > x/\log(x)$ for $x \geq 17$ \cite{Rosser1962ApproximateFF}. Note that the original statement uses the natural logarithm. But due to the division, the estimation also holds for the logarithm with base 2. Setting $x = y^2$, it holds that $\pi(y^2) > y^2/\log(y^2)$ for $y \geq 5$. As $y^2/\log(y^2) = y^2/(2\log(y)) \geq y^2/y = y$ for $y \geq 5$, it also holds that $\pi(y^2) > y$ for $y \geq 5$. Thus $p_i < i^2$ for $i \geq 5$, as we have at least $i$ many prime numbers in the interval~$[1, i^2]$. 

Manually checking the values for the first four prime numbers shows that the equation $p_i \leq i^2$ even holds for all prime numbers greater $2$. For $p_1 = 2 > 1^2$, we can simply multiply an additional factor of $2$. Altogether, we can thus estimate the product of the first $k$ prime numbers for $k \geq 2$ as

\begin{align*}
\prod_{i=1}^k q_i \leq \prod_{i=1}^k (i^2) \cdot 2 = (\prod_{i=1}^k i)^2 \cdot 2 = (k!)^2 \cdot 2 \leq (2 (k/2)^k)^2 \cdot 2 \\
= 2^2 ((k/2)^k)^2 \cdot 2 = 2^3 (k/2)^{2k} = 2^3 2^{2k \log(k/2)} \leq 2^{2k \log(k)}
\end{align*}
proving the statement. We use the estimation $k! = 2(k/2)^k$ which can easily be proved using induction. Further, note that $k \geq 2$ has to hold for the last estimation. 
\end{proof}

\begin{theorem}\label{t:3SattoQCRed}
The \QC{} problem is NP-hard even if the prime factorization of $\beta$ is given and each prime factor greater than $2$ occurs at most once and the prime factor $2$ occurs $4$ times.
\end{theorem}
\begin{proof}
We show a reduction from the well-known NP-hard problem $3$-\textsc{SAT} where we are given a $3$-SAT formula $\Phi$ with $n_\SatInd$ variables and $m_\SatInd$ clauses. \\

\noindent\textit{Transformation:} 
First, eliminate duplicate clauses from $\Phi$ and those where some variable $x_i$ and its negation $\bar{x_i}$ appear together. Call the resulting formula $\Phi'$, the number of occurring variables $n'$ and denote by $m'$ the number of appearing clauses respectively. Let $\Sigma = (\sigma_1, \dots, \sigma_{m'})$ be some enumeration of the clauses. Denote by $p_0, \dots, p_{2m'}$ the first $2m'+1$ prime numbers. Compute
\begin{align*}
\tau_{\Phi'} = - \sum_{i=1}^{m'} \prod_{j = 1}^{i} p_j .
\end{align*}  

Further, compute for each $i \in 1, 2, \dots, n'$:
\begin{align*}
f_i^+ = \sum_{x_i \in \sigma_j} \prod_{k = 1}^{j} p_k \text{\,\,\, and\,\,\,}
f_i^- = \sum_{\bar{x_i} \in \sigma_j} \prod_{k = 1}^{j} p_k .\\
\end{align*}
 
Set $\nu = 2m'+n'$. Compute the coefficients $c_j$ for all $j = 0, 1, \dots, \nu$ as follows: Set $c_0 = 0$. For $j = 1, \dots, 2m'$ set
\begin{align*}
c_j = -\frac{1}{2} \prod_{i = 1}^{j} p_i \text{\,\,\, for\,} j = 2k-1 \text{\,\,\, and\,\,\,}
c_j = - \prod_{i = 1}^{j} p_i \text{\,\,\, for\,} j = 2k.\\
\end{align*}

Compute the remaining ones for $j = 1, \dots n'$ as $c_{2m'+j} = 1/2 \cdot (f_j^+ - f_j^-)$. Further, set $\tau = \tau_{\Phi'} + \sum_{j=0}^{\nu} c_j + \sum_{i=1}^{n'} f_i^-$. 

Denote by $q_1, \dots, q_{\nu^2+2\nu+1}$ the first $\nu^2+2\nu+1$ prime numbers. Let $p_{0,0}, p_{0,1}, \dots,$ $p_{0,\nu},  p_{1,0}, \dots,  p_{\nu,\nu}$ be the first $(\nu+1)^2 = \nu^2 + 2\nu + 1$ prime numbers greater than $(4(\nu+1) 2^3 \prod_{i=1}^{\nu^2+2\nu+1} q_i)^{1/((\nu^2 + 2\nu + 1)\log(\nu^2+2\nu+1))}$ and greater than $p_{2m'}$.  Define $p^*$ as the $(\nu^2+2\nu+2m'+13)$th prime number.

Determine the parameters $\theta_j$ for $j = 0, 1, \dots, \nu$ as the least $\theta_j$ satisfying:
\begin{align*}
    & \theta_j \equiv c_j \bmod 2^3 \cdot p^* \prod_{i=1}^{m'} p_i, \\
    & \theta_j \equiv 0 \bmod \prod_{i = 0, i\neq j}^{\nu} \prod_{k=0}^{\nu}p_{i,k}, \\
    & \theta_j \not\equiv 0 \bmod p_{j,1} .
\end{align*}

Set the following parameters:
\begin{align*}
    & H = \sum_{j=0}^{\nu} \theta_j
    \text{\,\,\,and\,\,\,} K = \prod_{i = 0}^{\nu} \prod_{k=0}^{\nu}p_{i,j}  .
\end{align*}
Finally, set
\begin{align*}
    & \alpha = (2^4 \cdot p^* \prod_{i=1}^{m'} p_i + K)^{-1} \cdot (K\tau^2 + 2^4 \cdot p^* \prod_{i=1}^{m'} p_i\cdot H^2),\\
    & \beta = 2^4 \cdot p^* \prod_{i=1}^{m'} p_i \cdot K,\\
    & \gamma = H.
\end{align*}
where $(2^4 \cdot p^* \prod_{i=1}^{m'} p_i + K)^{-1}$ is the inverse of $(2^4 \cdot p^* \prod_{i=1}^{m'} p_i + K) \bmod 2^4 \cdot p^* \prod_{i=1}^{m'} p_i \cdot K$. \\

\noindent\textit{Correctness:}
We show that the satisfiability of the formula~$\Phi$ is equivalent to a line of (systems of) equations, \ie, the formula has a satisfying truth assignment on the variables if and only if the  (systems of) equations admit a solution.
By this, we prove the hardness for various problems along the way. These are listed above with their respective equivalence sketched. In the following, we separate each of these steps by claims.

However, before we start with the transformations of the formula, we first observe some properties about the generated prime factors. These come in handy for the estimations later on. In particular, we want to show that choosing $p^*$ as the $(\nu^2+2\nu+2m'+13)$th prime factor satisfies $p^* > p_{\nu,\nu}$: Suppose $p_{2m'} \geq (4(\nu+1) 2^3 \cdot \prod_{i=1}^{\nu^2+2\nu+1} q_i)^{1/((\nu^2 + 2\nu + 1)\log(\nu^2+2\nu+1))}$. Then $p_{\nu,\nu}$ is the $(\nu^2+2\nu+1+2m'+1)$th prime number and thus $p^* > p_{\nu,\nu}$. Otherwise, if $p_{2m'} < (4(\nu+1) 2^3 \prod_{i=1}^{\nu^2+2\nu+1} q_i)^{1/((\nu^2 + 2\nu + 1)\log(\nu^2+2\nu+1))}$, we bound the function values as follows:

\begin{align*}
& (4(\nu+1) 2^3 \prod_{i=1}^{\nu^2+2\nu+1} q_i)^{1/((\nu^2 + 2\nu + 1)\log(\nu^2+2\nu+1))} \\
& = 4^{1/((\nu^2 + 2\nu + 1)\log(\nu^2+2\nu+1))} (\nu+1)^{1/((\nu^2 + 2\nu + 1)\log(\nu^2+2\nu+1))}\\
& \text{\,\,\,\,\,\,\,\,}\cdot (2^3)^{1/((\nu^2 + 2\nu + 1)\log(\nu^2+2\nu+1))}\\
& \text{\,\,\,\,\,\,\,\,}\cdot (\prod_{i=1}^{\nu^2+2\nu+1} q_i)^{1/((\nu^2 + 2\nu + 1)\log(\nu^2+2\nu+1))}\\
& \leq 2 \cdot 2 \cdot 2 \cdot (2^{2(\nu^2 + 2\nu + 1)\log(\nu^2+2\nu+1)})^{1/((\nu^2 + 2\nu + 1)\log(\nu^2+2\nu+1))} \\
& \leq 8 \cdot (4^{(\nu^2 + 2\nu + 1)\log(\nu^2+2\nu+1)})^{1/((\nu^2 + 2\nu + 1)\log(\nu^2+2\nu+1))} \\
& = 8 \cdot 4
 = 32 .
\end{align*}
The second transformation holds as the product of the first $k$ prime numbers is bounded by $2^{2k\log(k)}$ (for $k \geq 2$, which obviously holds here), see Lemma~\ref{l:ProdPrimes}. There are 11 prime numbers in the interval $[1,32]$. Thus, $p_{\nu,\nu}$ is at most the $(11+\nu^2+2\nu+1)$th prime number and thus $p^* > p_{\nu,\nu}$. 

Further, note that $p^* \leq \prod_{i=m'+1}^{\nu^2+2\nu+1} q_i$: 
We can bound the value of the product from beneath as $\prod_{i=m'+1}^{\nu^2+2\nu+1} q_i \geq q_{m'+1}^{\nu^2+\nu}$.
Estimating the value for $p^*$, we use that the value of the next prime number after a number $\rho$ is at most $2\rho$ \cite{bertrand2018bertrand}. Thus, as there are $\nu^2+2\nu+m'+11$ prime numbers between $p_{m'+1}$ and $p^*$, we get $p^* \leq q_{m'+1} \cdot 2^{\nu^2+2\nu+m'+11} \leq  q_{m'+1} \cdot 2^{\nu^2+3\nu+11}$ since per definition $\nu \geq m'$ holds.
Dividing both sides of the estimation by $q_{m'+1}$, it thus remains to show that $2^{\nu^2+3\nu+11} \leq  q_{m'+1}^{\nu^2+\nu-1}$. 
Obviously, $q_{m'+1}^{\nu^2+\nu-1}$ grows for larger values of $m'$. The smallest reasonable value for $m' = 2$ and thus $q_{m'+1} \geq 5$. By that, we get that 
\begin{align*}
q_{m'+1}^{\nu^2+\nu-1} \geq 5^{\nu^2+\nu-1} \geq 2^{2(\nu^2+\nu-1)} = 2^{2\nu^2+2\nu-2} \geq 2^{\nu^2+3\nu+11}
\end{align*}
for all $\nu \geq 5$ and thus for all reasonable values of $\nu$, showing the statement.

Let us now focus on the transformations of the formula $\Phi$ yielding the first claim:

\begin{claim}\label{claim:diss:SatToP2}
The $3$-\textsc{SAT} problem asking whether there is a truth assignment $\eta$ that satisfies all clauses $\sigma_k$ of the $3$-SAT formula $\Phi$ simultaneously is a yes-instance if and only if Problem (P2) asking whether there are values $y_k \in \{0,1,2,3\}$ and a truth assignment $\eta$ such that $0 = R_k = y_k - \sum_{x_i \in \sigma_k} \eta(x_i) - \sum_{\bar{x_i} \in \sigma_k} (1 - \eta(x_i)) +1$ for all $k$ is a yes-instance.
\end{claim}
\begin{proof}
Obviously, the reduced formula $\Phi'$ is satisfiable if and only if $\Phi$ is. The formula~$\Phi'$ is satisfiable if there exists a truth assignment $\eta \colon \{x_1, \dots, x_{n'}\} \rightarrow \{0,1\}$ assigning a logical value to each variable~$x_1, \dots, x_{n'}$ which satisfies all clauses $\sigma_1, \dots, \sigma_{m'}$ simultaneously. This can be re-written to the following equation for each clause $\sigma_k \in \Phi_k$ interpreting the truth values as numbers:
\begin{align*}
& 0 = R_k = y_k - \sum_{x_i \in \sigma_k} \eta(x_i) - \sum_{\bar{x_i} \in \sigma_k} (1 - \eta(x_i)) +1 \text{,\,\,\,} y_k \in \{0,1,2,3\}.
\end{align*}

For a clause $\sigma_k$, this equation is only satisfiable if at least one variable $x_i \in \sigma_k$ has value $\eta(x_i) = 1$ or one variable occurring in its negation $\bar{x_i} \in \sigma_k$ has value $\eta(x_i) = 0$. Otherwise, we have to set $y_k = -1$ which is not allowed.
\end{proof}

Note that we never have to set $y_k = 3$ to satisfy the formula. However, we allow this value as it will come in handy later on when transforming the equation. Further, set $0 = R_0 = \alpha_0 + 1$ for $\alpha_0 \in \{-1, +1\}$ for later convenience. Clearly, the new equation is satisfiable. 

\begin{claim}\label{claim:diss:P2ToP3}
The Problem (P2) asking whether there are values $y_k \in \{0,1,2,3\}$ and a truth assignment $\eta$ such that $0 = R_k = y_k - \sum_{x_i \in \sigma_k} \eta(x_i) - \sum_{\bar{x_i} \in \sigma_k} (1 - \eta(x_i)) +1$ for all $k$ is a yes-instance if and only if Problem (P3) asking whether there are values $\alpha_j \in \{-1, +1\}$ such that $\sum_{j=0}^\nu \theta_j \alpha_j \equiv \tau \bmod 2^3 \cdot p^* \prod_{i=1}^{m'} p_i$ is a yes-instance.
\end{claim}
\begin{proof}
We can bound the values of $R_k$ for $k \in \{0, 1, \dots, m'\}$ by $-2 \leq R_k \leq 4$. For the lower bound, the values are given by $y_k = 0$, all $x_i \in \sigma_k$ have value $\eta(x_i) = 1$ and all $\bar{x_i} \in \sigma_k$ have value $\eta(x_i) = 0$. For the upper bound we set $y_k = 3$, all $x_i \in \sigma_k$ to $\eta(x_i) = 0$ and $\bar{x_i} \in \sigma_k$ to $\eta(x_i) = 1$. For $R_0$ obviously $0 \leq R_k \leq 2$ holds. Thus,
\begin{align*}
R_k = 0 \text{,\,} \forall k \in \{0, 1, \dots, m'\} \Leftrightarrow \sum_{k=0}^{m'} R_k \prod_{i=0}^k p_i = 0
\end{align*}
as the sum is zero if all $R_k = 0$. For the opposite direction, if the sum is zero, then no $R_k \neq 0$
as the product of the prime numbers grows too fast. Thus, the other summands cannot compensate for some $R_k \neq 0$. We can bound the expression further by
\begin{align*}
|\sum_{k=0}^{m'} R_k \prod_{i=0}^k p_i| \leq 4 \sum_{k=0}^{m'} \prod_{i=0}^k p_i \leq 4 (m'+1) \prod_{i=0}^{m'} p_i < 2^3 \cdot p^* \prod_{i=1}^{m'} p_i  
\end{align*}
as $p^* >  p_{\nu,\nu} > p_{m'} > m'+1$. This yields 
\begin{align*}
\tag{I}
R_k = 0 \text{,\,} \forall k \in \{0, 1, \dots, m'\} \Leftrightarrow \sum_{k=0}^{m'} R_k \prod_{i=0}^k p_i \equiv 0 \bmod 2^3 \cdot p^* \prod_{i=1}^{m'} p_i 
\end{align*}
as the modulo has no impact on the satisfiability of the equation.

Next, we aim to re-write $R_k$ by replacing the variables $y_k$ and $\eta(x_i)$ with new variables admitting a domain of $\{-1, 1\}$:
\begin{align*}
& y_k = 1/2 \cdot [ (1-\alpha_{2k-1}) + 2 \cdot (1-\alpha_{2k}) ] \text{,\,\,\,} k \in \{1, \dots, m'\}, \\
& \eta(x_i) = 1/2 \cdot (1 - \alpha_{2m'+i}) \text{,\,\,\,} i \in \{1, \dots, n'\}.
\end{align*}

Obviously the value domains of $y_k$ and $\eta(x_i)$ are preserved. Substituting the variables and re-arranging the equation (I) yields

\begin{align*}
\sum_{j=0}^{\nu} c_j \alpha_j \equiv \tau \bmod 2^3 \cdot p^* \prod_{i=1}^{m'} p_i \text{, \,} \alpha_j \in \{-1, +1\}.
\end{align*}
By definition of $\theta_j$ this is equivalent to
\begin{align*}
\sum_{j=0}^\nu \theta_j \alpha_j \equiv \tau \bmod 2^3 \cdot p^* \prod_{i=1}^{m'} p_i  \text{, \,} \alpha_j \in \{-1, +1\}
\end{align*}
proving the claim.
\end{proof}

Let $H = \sum_{j=0}^{\nu} \theta_j$ and $K = \prod_{i = 0}^{\nu} \prod_{j=0}^{\nu}p_{i,j}$ be defined as before. Consider the following system asking whether there is a $x \in \mathbb{Z}$ such that:
\begin{align*}
\tag{P4.1}
& 0 \leq |x| \leq H\\
\tag{P4.2}
& (H + x) (H -x) \equiv 0 \bmod K
\end{align*}
We use this system to integrate the condition $x \leq H$ into the transformations. In the following, we prove that each solution of this system is of form $x = \sum_{j=0}^\nu \alpha_j \theta_j$ and thus Problem (P4) can be combined with Problem (P3) yielding Problem~(P5).

\begin{claim}\label{claim:diss:P3ToP5}
The Problem (P3) asking whether there are values $\alpha_j \in \{-1, +1\}$ such that $\sum_{j=0}^\nu \theta_j \alpha_j \equiv \tau \bmod 2^3 \cdot p^* \prod_{i=1}^{m'} p_i$ is a yes-instance if and only if the Problem~(P5) is a yes-instance.
\end{claim}
\begin{proof}
The unique solutions $x$ to the given system (P4) are of form
\begin{align*}
x = \sum_{j=0}^\nu \alpha_j \theta_j \text{, \,} \alpha \in \{-1, +1\} \text{, \,} j = 0, 1, \dots, \nu.
\end{align*}
Let us first verify that an $x$ of such form solves the system. First 
\begin{align*}
|x| = |\sum_{j=0}^\nu \alpha_j \theta_j | \leq \sum_{j=0}^\nu \theta_j = H
\end{align*}
satisfies (P4.1). Further, we have that each summand in the expanded formula~$(H+x)(H-x)$ has to contain all prime factors $p_{i,j}$ for $i = 0, 1, \dots, \nu$ and $j = 0, 1, \dots, \nu$ in its prime factorization to satisfy (P4.2). For $(H+x) = (\sum_{j=0}^{n} \theta_j + \sum_{j=0}^n \theta_j \alpha_j)$ it holds that each $\theta_j$ where $\alpha_j = +1$ occurs twice while each $\theta_j$ where $\alpha_j = -1$ is canceled out by $H$. The other way round holds for $(H-x)$. Thus, expanding the brackets yields that each summand is a product of some $\theta_j$ and $\theta_k$ where $\alpha_j = +1$ and $\alpha_k = -1$. This implies that $j \neq k$. As each $\theta_j$ contains all prime factors of $K$ except $p_{j,0}, \dots, p_{j, \nu}$, the product of two different $\theta_j$ and $\theta_k$ contains each prime factor occurring in $K$ satisfying (P4.2). 

Regarding the uniqueness, observe that 
\begin{align*}
(H + x) (H -x) \equiv 0 \bmod \prod_{j=0}^\nu p_{i,j} \text{, \,} \forall i = 0, 1, \dots, \nu .
\end{align*}
Assume there exists some number $\tilde{p} = \prod_{j=0}^\nu p_{i,j}$ for some $i \in \{0, 1, \dots, \nu\}$ which divides $(H+x)$ and $(H-x)$ (without remainder). Thus, $(H+x)+(H-x) \equiv 0 \bmod \tilde{p} \Leftrightarrow 2H \equiv 0 \bmod \tilde{p}$. As $\tilde{p}$ is a product of prime numbers greater than $2$ is follows that $H \equiv 0 \bmod \tilde{p} \Leftrightarrow \sum_{j=0}^{\nu} \theta_j \equiv 0 \bmod \tilde{p}$. However, from the definition of $\theta_j$ (third condition) it follows that for each $j$ there exist different prime numbers not present in the prime factorization of $\theta_j$ contradicting the assumption. Thus, $\tilde{p}$ divides either $(H+x)$ or $(H-x)$ (without remainder). Define
\begin{align*}
& \alpha_i = 
\begin{cases}
    +1 & \text{ if \,} (H-x) \equiv 0 \bmod \prod_{j=0}^\nu p_{i,j}\\
    -1 & \text{ if \,} (H+x) \equiv 0 \bmod \prod_{j=0}^\nu p_{i,j}
\end{cases} \\
 &   x' = \sum_{i=0}^\nu \alpha_i \theta_i .
\end{align*}

In the following, we show that $x' \equiv x \bmod  \prod_{j=0}^\nu p_{i,j}$ holds:
\begin{align*}
& x' \equiv x \bmod  \prod_{j=0}^\nu p_{i,j} \\
& \Leftrightarrow \sum_{i=0}^\nu \alpha_i \theta_i \equiv x \bmod \prod_{j=0}^\nu p_{i,j} \\
& \Leftrightarrow \alpha_i \theta_i \equiv x \bmod \prod_{j=0}^\nu p_{i,j}
\end{align*}
\begin{align*}
& \Leftrightarrow \sum_{k=0}^\nu \alpha_i \theta_k \equiv x \bmod \prod_{j=0}^\nu p_{i,j}\\
& \Leftrightarrow  \alpha_i \sum_{k=0}^\nu \theta_k \equiv x \bmod \prod_{j=0}^\nu p_{i,j}\\
& \Leftrightarrow  \alpha_i H \equiv x \bmod \prod_{j=0}^\nu p_{i,j}
\end{align*}
for all $i \in \{0, 1, \dots, \nu\}$. The first transformation simply inserts the definition of $x'$. Due to the definition of the $\theta_k$, only the summand $\theta_i$ remains after calculating the modulo. Thus, we can sum up all $\theta_k$ with arbitrary sign as they equal zero after calculating the modulo. In the last step we insert the definition of $H$. Now we either have $\alpha_j = +1$. Then $H \equiv x \bmod \prod_{j=0}^\nu p_{i,j}$, \ie, $H - x\equiv 0 \bmod \prod_{j=0}^\nu p_{i,j}$, which is true by definition of $\alpha_j = +1$. Otherwise, $\alpha_j = -1$. Then $-H \equiv x \bmod \prod_{j=0}^\nu p_{i,j}$, \ie, $H + x \equiv 0 \bmod \prod_{j=0}^\nu p_{i,j}$, which is again true by the definition of $\alpha_j$. Thus, the initial statement is correct.

So, as $\alpha_j \in \{-1,+1\}$ for all $j \in \{0, 1, \dots, \nu\}$, it holds that $-H \leq x \leq H$. Since the same holds for $x'$ it follows that $|x-x'| \leq 2H$. Let us bound the value of $\lambda_j = \theta_j / (\prod_{i = 0, i\neq j}^{\nu} \prod_{k=0}^{\nu}p_{i,k})$. It holds that $\theta_j \leq  2^4 \cdot p^* \prod_{i=1}^{m'} p_i \cdot \prod_{i = 0, i\neq j}^{\nu} \prod_{k=0}^{\nu}p_{i,k}$, as $2^3 \cdot p^* \prod_{i=1}^{m'} p_i$ and $\prod_{i = 0, i\neq j}^{\nu} \prod_{k=0}^{\nu}p_{i,k}$ are coprime and thus the least $\theta_j$ satisfying the equivalence conditions in the definition of $\theta_j$ is at most their product~\cite{schroeder2009chinese}.
The additional factor of $2$ is introduced by the inequality constraint $\theta_j \not\equiv 0 \bmod p_{j,1}$, as if the calculated $\theta_j$ for the equality constraints does not satisfy that condition, we can extend it to $\theta_j' = \theta_j + 2^3 \cdot p^* \prod_{i=1}^{m'} p_i \cdot \prod_{i = 0, i\neq j}^{\nu} \prod_{k=0}^{\nu}p_{i,k}$. This doubles the size estimation and as $p_{j,1}$ is coprime to $2^3 \cdot p^* \prod_{i=1}^{m'} p_i \cdot \prod_{i = 0, i\neq j}^{\nu} \prod_{k=0}^{\nu}p_{i,k}$, it holds that $\theta_j'$ is not equivalent to $0 \bmod p_{j,1}$.
Thus,
\begin{align*}
& \lambda_j = \frac{\theta_j}{\prod_{i = 0, i\neq j}^{\nu} \prod_{k=0}^{\nu}p_{i,k}} \\
& \leq \frac{2^4 \cdot p^* \prod_{i=1}^{m'} p_i \cdot \prod_{i = 0, i\neq j}^{\nu} \prod_{k=0}^{\nu}p_{i,k}}{\prod_{i = 0, i\neq j}^{\nu} \prod_{k=0}^{\nu}p_{i,k}} \\
& = \frac{2^4 \cdot p^* \prod_{i=1}^{m'} p_i \cdot K/(\prod_{k=0}^\nu p_{j, k})}{\prod_{i = 0, i\neq j}^{\nu} \prod_{k=0}^{\nu}p_{i,k}}\\
& = \frac{2^4 \cdot p^* \prod_{i=1}^{m'} p_i \cdot K}{\prod_{i = 0}^{\nu} \prod_{k=0}^{\nu}p_{i,k}} \\
& \leq \frac{2^4 \cdot p^* \prod_{i=1}^{m'} p_i \cdot K}{4(\nu+1)2^3 \cdot p^* \prod_{i=1}^{m'} p_i} \\
& = \frac{K}{2(\nu+1)}.
\end{align*}
To validate the fourth estimation, we have to prove that $\prod_{i = 0}^{\nu} \prod_{k=0}^{\nu}p_{i,k} \geq 4(\nu+1)2^3 \cdot p^* \prod_{i=1}^{m'} p_i$. As previously shown, it holds that $p^* \leq \prod_{i=m'+1}^{\nu^2+2\nu+1} q_i$. Thus, 
\begin{align*}
4(\nu+1)2^3 \cdot p^* \prod_{i=1}^{m'} p_i \leq 4(\nu+1)2^3 \prod_{i=m'+1}^{\nu^2+2\nu+1} q_i \prod_{i=1}^{m'} p_i = 4(\nu+1)2^3 \prod_{i=1}^{\nu^2+2\nu+1} q_i .
\end{align*}
Hence, it remains to prove that $\prod_{i = 0}^{\nu} \prod_{k=0}^{\nu}p_{i,k} \geq 4(\nu+1)2^3 \prod_{i=1}^{\nu^2+2\nu+1} q_i$. Per definition,  $p_{0,0} > \max\{p_{2m'}, q_{11}\}$. Thus, comparing the factors of both products, we see that $ 4(\nu+1)2^3 \prod_{i=1}^{\nu^2+2\nu+1} q_i$ has $4(\nu+1)2^3$ and the first $\max\{2m', 11\}$ prime numbers smaller than $p_{0,0}$ uniquely, whereas $\prod_{i = 0}^{\nu} \prod_{k=0}^{\nu}p_{i,k}$ has the largest $\max\{2m', 11\}$ prime factors uniquely. 
Let us consider the smallest case where $\max\{2m', 11\} = 4$ as the smallest reasonable value for $m' = 2$ (a formula with just one clause is trivial). The smallest reasonable value for $\nu = 7$ if $m' = 2$ and $n'=3$ (less than 3 variables is not possible).
Now it is easy to prove via manual calculation that the product of $4(\nu+1)2^3$ times the first $4$ prime numbers (smaller than $p_{0,0}$) is indeed smaller than the product of the next $4$ prime numbers larger than $q_{\nu^2+2\nu+1} = q_{65}$.
If $m'$ grows, we get the same number of additional prime factors for both products, whereas each new prime number in $\prod_{i = 0}^{\nu} \prod_{k=0}^{\nu}p_{i,k}$ is larger than the additional ones in the other product. If we have larger values for $\nu$, it only influences the product $4(\nu+1)2^3 \cdot p^* \prod_{i=1}^{m'} p_i$ linearly, whereas for the other product, we start with greater prime numbers, thus having a larger impact on the product. Hence, the estimation is correct for all values.

The term $K/(2(\nu+1))$ bounds each summand of $H$ as it considers their largest value to satisfies the constraints as well as the modulo when calculating the values (see definition of $\theta_k$). It follows that $2H = 2 \sum_{j=0}^\nu \theta_j < 2 \cdot (\nu+1) \cdot K/(2(\nu+1)) = K$. Thus, $x = x'$ and we conclude that solution of the form $x=\sum_{j=0}^\nu \theta_j \alpha_j$ are the unique solutions to the system (P4.1) and (P4.2).

Thus, we can re-write 
\begin{align*}
\sum_{j=0}^\nu \theta_j \alpha_j \equiv \tau \bmod 2^3 \cdot p^* \prod_{i=1}^{m'} p_i  \text{, \,} \alpha_j \in \{-1, +1\}
\end{align*}
using the system (P4.1) and (P4.2) to the following one:
\begin{align*}
\tag{P5.1}
& 0 \leq |x| \leq H \text{, \,} x \in \mathbb{Z}\\
\tag{P5.2}
& x \equiv \tau \bmod 2^3 \cdot p^* \prod_{i=1}^{m'} p_i \\
\tag{P5.3}
& (H+x)(H-x) \equiv 0 \bmod K 
\end{align*}
proving their equivalence.
\end{proof}

\noindent Next, we re-write the system (P5) to:
\begin{align*}
\tag{P6.1}
& 0 \leq |x| \leq H \text{, \,} x \in \mathbb{Z}\\
\tag{P6.2}
& (\tau-x)(\tau+x) \equiv 0 \bmod 2^4 \cdot p^* \prod_{i=1}^{m'} p_i \\
\tag{P6.3}
& (H+x)(H-x) \equiv 0 \bmod K .
\end{align*}

\begin{claim}\label{claim:diss:P5ToP6}
The Problem (P5) is a yes-instance if and only if the Problem (P6) is a yes-instance.
\end{claim}
\begin{proof}
As only the second conditions differ, we focus on their equivalence in the following.
First, we prove that if (P5.2) holds, \ie, $x \equiv \tau \bmod 2^3 \cdot p^* \prod_{i=1}^{m'} p_i$, then (P6.2) holds, \ie, $(\tau-x)(\tau+x) \equiv 0 \bmod 2^4 \cdot p^* \prod_{i=1}^{m'} p_i$. We can re-write (P5.2) to $x = \lambda 2^3 \cdot p^* \prod_{i=1}^{m'} p_i + \tau$ for some $\lambda \in \mathbb{Z}$. Inserting this in (P6.2) yields:
\begin{align*}
&(\tau + \lambda 2^3 \cdot p^* \prod_{i=1}^{m'} p_i + \tau)(\tau - \lambda 2^3 \cdot p^* \prod_{i=1}^{m'} p_i - \tau) \\
&= (2\tau + \lambda 2^3 \cdot p^* \prod_{i=1}^{m'} p_i) (\lambda 2^3 \cdot p^* \prod_{i=1}^{m'} p_i) \equiv 0 \bmod 2^3 \cdot p^* \prod_{i=1}^{m'} p_i
\end{align*}
as each factor is multiplied with $\lambda 2^3 \cdot p^* \prod_{i=1}^{m'} p_i$.

Next, we prove the opposite direction. First, observe that if $(\tau-x)(\tau+x) \equiv 0 \bmod 2^4 \cdot p^* \prod_{i=1}^{m'} p_i$ then either $(\tau-x) \equiv 0 \bmod 2^3$ or $(\tau+x) \equiv 0 \mod 2^3$: As (P5.2) holds, $(\tau+x) = \lambda_i \cdot 2^i$ and $(\tau-x) = \lambda_j \cdot 2^j$ for some $i,j \in \mathbb{Z}$ and $\lambda_i, \lambda_j \not\equiv 0 \bmod 2$. It follows that

\begin{align*}
& (\tau+x)+(\tau-x) = \lambda_i \cdot 2^i + \lambda_j \cdot 2^j \\
& \Leftrightarrow 2 \tau = \lambda_i \cdot 2^i + \lambda_j \cdot 2^j \\
& \Leftrightarrow \tau = \lambda_i \cdot 2^{i-1} + \lambda_j \cdot 2^{j-1}.
\end{align*}
As $\tau$ is odd per definition, either $i$ or $j$ has to be $1$ and thus the other parameter has to be $3$. 
Using this, we know that if $x$ satisfies $(P5.2)$, then $(\tau-x) \equiv 0 \bmod 2^3$ or $(\tau+x) \equiv 0 \mod 2^3$. In the first case, $x$ directly corresponds to a solution of (P6.2) as $x-\tau$ is a multiple of $2^3$ and thus $x$ is a multiple of $2^3$ with a residue of $\tau$. Otherwise $-x$ satisfies the condition using the same argument. Obviously the other conditions are also satisfied in both systems. 
\end{proof}

Lastly, we re-write the system one final time to:
\begin{align*}
\tag{QC.1}
& 0 \leq x \leq H \text{, \,} x \in \mathbb{Z}\\
& 2^4 \cdot p^* \cdot  \prod_{i=1}^{m'} p_i (H^2 - x^2) + K (\tau^2 - x^2)\\
\tag{QC.2}
& \equiv 0 \bmod 2^4 \cdot p^* \cdot  \prod_{i=1}^{m'} p_i \cdot K .
\end{align*}

\begin{claim}\label{claim:diss:P6ToQC}
The Problem (P6) is a yes-instance if and only if the \QC{} problem is a yes-instance.
\end{claim}
\begin{proof}
First, as we only consider $x^2$, we can suppose $x \geq 0$ and thus re-writting (P6.1) to (QC.1) is correct. Further, (P6.2) and (P6.3) merge into (QC.2). Recall that $2^4 \cdot p^* \cdot  \prod_{i=1}^{m'} p_i$ and $K$ are co-prime. The first summand obviously always contains the factor $2^4 \cdot p^* \cdot  \prod_{i=1}^{m'} p_i$, thus we have to find an $x$ such that $(H^2 - x^2) \equiv 0 \bmod K$ which corresponds to (P6.3). The second summand clearly is a multiple of $K$, thus we have to assure that $(\tau^2 - x^2) \equiv 0 \bmod 2^4 \cdot p^* \cdot  \prod_{i=1}^{m'} p_i$. This matches (P5.2).

Dissolving the brackets and rearranging the term (QC.2) we get
\begin{align*}
&(2^4 \cdot p^* \cdot  \prod_{i=1}^{m'} p_i + K)x^2\\
& \equiv K\tau^2 + 2^4 \cdot p^* \cdot  \prod_{i=1}^{m'} p_i H^2 \bmod 2^4 \cdot p^* \cdot  \prod_{i=1}^{m'} p_i \cdot K .
\end{align*}
As $2^4 \cdot p^* \cdot  \prod_{i=1}^{m'} p_i + K$ is relatively prime to $2^4 \cdot p^* \cdot  \prod_{i=1}^{m'} p_i \cdot K$ it has an inverse modulo $2^4 \cdot p^* \cdot  \prod_{i=1}^{m'} p_i \cdot K$ \cite{lame1844note}. Thus, multiplying by the inverse we get the values for $\alpha, \beta$ and $\gamma$ as in the transformation above.
\end{proof}

Overall, this proves that satisfying the formula $\Phi$ is equivalent to an instance of the \QC{} problem admitting a feasible solution.\\

\noindent\textit{Running time:}
All steps, numbers and their computation can be bounded in a polynomial dependent of $n_\SatInd$, \ie, the number of variables in the $3$-Sat formula, and $m_\SatInd$, \ie, the number of clauses in the formula. 
First, we eliminate unnecessary clauses from the formula. Thus, we have to go through all clauses once. The first $2m'+1$ prime numbers have a value of at most $O(m' \log(m'))$ and can thus be found in polynomial time via sieving. The function $(4(\nu+1) 2^3 \prod_{i=1}^{\nu^2+2\nu+1} q_i)^{1/((\nu^2+2\nu+1)\log(\nu^2 + 2\nu + 1))}$ is at most $32$ as shown before. Thus, we can also bound the value of the next $\nu^2 + 2\nu + 1$ prime numbers larger than $32$ and $p_{2m'}$ by a polynomial in $n_\SatInd$ and $m_\SatInd$ and we can compute them efficiently by sieving. All other numbers calculated in the transformation are a product or sum over these prime numbers (each occurring at most once in the calculation) and thus their values are also in poly$(n_\SatInd, m_\SatInd)$. We can compute the inverse $(2^4 \cdot p^* \prod_{i=1}^{m'} p_i + K)^{-1}$ in polynomial time \cite{lame1844note}.
\end{proof}

Now we have proved that the \QC{} problem is NP-hard even in the restricted case where all prime factors in $\beta$ only appear at most once (except 2). To apply the ETH, however, we also have to estimate the dimensions of the generated instance. Denote by $B = b_1^{\beta_1}, \dots, b_{n_{\QCInd}}^{\beta_{n_{\QCInd}}}$ the prime factorization of $\beta$ where $b_1, \dots, b_{n_{\QCInd}}$ denotes the different prime factors of $\beta$ and $\beta_i$ the occurrence of $b_i$. The above reduction yields the following parameters:

\begin{theorem}\label{t:3SattoQCInstanceSize}
An instance of the $3$-\textsc{SAT} problem with $n_\SatInd$ variables and $m_\SatInd$ clauses is reducible to an instance of the \QC{} problem in polynomial time with the properties that  $\alpha, \beta, \gamma \in 2^{O((n_\SatInd+m_\SatInd)^2 \log(n_\SatInd + m_\SatInd))}$, $n_{\QCInd} \in O((n_\SatInd+m_\SatInd)^2)$, $max_{i}\{b_i\} \in O((n_\SatInd+m_\SatInd)^2\log(n_\SatInd+m_\SatInd))$, and each prime factor in $\beta$ occurs at most once except the prime factor $2$ which occurs four times.
\end{theorem}
\begin{proof}
In Theorem \ref{t:3SattoQCRed}, we already showed and proved a reduction from the $3$-\textsc{SAT} problem to the \QC{} problem and argued the running time. It remains to bound the parameters. To do so, we bound the numbers occurring in the reduction above in order of their appearance.

After eliminating the trivial clauses it obviously holds that $m' \leq m_\SatInd$ and $n' \leq n_\SatInd$. Next, we calculate $\tau_{\Phi'}$. Its absolute value can be bounded as
\begin{align*}
&|\tau_{\Phi'}| = |- \sum_{i=1}^{m'} \prod_{j = 1}^{i} p_j| 
= \sum_{i=1}^{m'} \prod_{j = 1}^{i} p_j \\
&\leq m_\SatInd \prod_{j = 1}^{m_\SatInd} p_j 
\leq m_\SatInd 2^{2m_\SatInd \log(m_\SatInd)} \leq 2^{O(m_\SatInd \log(m_\SatInd))} 
\end{align*}
since the product of the first $k$ prime numbers is bounded by $2^{2k\log(k)}$ for all $k \geq 2$, see Lemma~\ref{l:ProdPrimes}. Similarly, $\max_i\{|f^+_i|, |f^-_i|\} \leq \sum_{x_i \in \sigma_j} \prod_{k = 1}^{j} p_k + \sum_{\bar{x_i} \in \sigma_j} \prod_{k = 1}^{j} p_k \leq 2m_\SatInd \cdot 2^{2m_\SatInd \log(m_\SatInd)}  \leq 2^{O(m_\SatInd \log(m_\SatInd))} $ and $\max_j\{c_j\} = \max_j\{\prod_{i=1}^j p_i, f^+_j + f^-_j\} \leq 2^{O(m_\SatInd \log(m_\SatInd))}$. Per definition, $\nu = 2m' + n' = O(n_\SatInd + m_\SatInd)$. 
The largest prime number $max_{i}\{b_i\}$ we generate in the reduction is $p^*$, which is the $(\nu^2+2\nu+2m'+13)$th prime number. Thus, its value is bounded by $p^* \leq O(\nu^2 \log(\nu)) = O((n_\SatInd+m_\SatInd)^2\log(n_\SatInd+m_\SatInd))$~\cite{hardy1916contributions}. 
Due to the modulo, we can bound $\max_j\{\theta_j\}$ as
\begin{align*}
&\max_j\{\theta_j\} \leq 2^4 \cdot p^* \prod_{i=1}^{m'} p_i \cdot \prod_{i = 0, i\neq j}^{\nu} \prod_{k=0}^{\nu} p_{i,k} \\
&\leq 2^4 2^{O((n_\SatInd+m_\SatInd)^2 \log(n_\SatInd + m_\SatInd))} = 2^{O((n_\SatInd+m_\SatInd)^2 \log(n_\SatInd + m_\SatInd))}.
\end{align*}
Thus, $H = \sum_{j=0}^\nu \theta_j \leq \nu \cdot 2^{O((n_\SatInd+m_\SatInd)^2 \log(n_\SatInd + m_\SatInd))} = 2^{O((n_\SatInd+m_\SatInd)^2 \log(n_\SatInd + m_\SatInd))}$ and $K = \prod_{i = 0}^{\nu} \prod_{k=0}^{\nu}p_{i,k} \leq2^{O((n_\SatInd+m_\SatInd)^2 \log(n_\SatInd + m_\SatInd))}$. Finally, we can bound the main parameters. As $\alpha$ is bounded by the modulo of $\beta$ is follows that $\alpha \leq \beta$. Further, $\beta = 2^4 \cdot p^* \prod_{i=1}^{m'} p_i \cdot K \leq 2^{O((n_\SatInd+m_\SatInd)^2 \log(n_\SatInd + m_\SatInd))}$. Per definition $\gamma = H$ and thus $\gamma \leq 2^{O((n_\SatInd+m_\SatInd)^2 \log(n_\SatInd + m_\SatInd))}$, which finalizes the estimation of the numbers. 
\end{proof}

\section{Reduction from the \QC{} problem}\label{sec:RedTo2Stage}
This sections presents the reduction from the \QC{} problem to the \TwoStageProb{} problem. First, we present a transformation of an instance of the \QC{} problem to an instance of the \Rest{} problem. This problem was not considered so far and serves as an intermediate step in this chapter. However, it might be of independent interest as it generalizes the prominent Chinese Remainder theorem.  Secondly, we show how an instance of the \Rest{} problem can be modeled as a \TwoStage{} ILP.
Recall that in the \Rest{} problem, we are given numbers $x_1, \dots, x_{n_{\MRD}}, y_1, \dots, y_{n_{\MRD}}, q_1, \dots, q_{n_{\MRD}}, \zeta \in \mathbb{N}$ where the $q_i$s are pairwise co-prime. The question is to decide whether there exists a natural number $z$ satisfying the following integer linear program and which is smaller or equal to $\zeta$:

\begin{align*}
& z \bmod q_1 \in  \{x_1, y_1\} \\    
& z \bmod q_2 \in  \{x_2, y_2\} \\  
& \vdots \\
& z \bmod q_{n_{\MRD}} \in \{x_{n_{\MRD}}, y_{n_{\MRD}}\} .
\end{align*}
In other words, we either should met the residue $x_i$ or $y_i$. Thus, we can re-write the equation as $z \equiv x_i \bmod q_i$ or $z \equiv y_i \bmod q_i$ for all $i$. 
Indeed, this problem becomes easy if $x_i = y_i$ for all $i$, \ie, we know the remainder we want to satisfy for each equation \cite{wagon1999mathematica}: First, compute $s_i$ and $r_i$ with $r_i \cdot q_i + s_i \cdot \prod_{j=1, j\neq i}^{n_{\MRD}} q_j = 1$ for all $i$ using the Extended Euclidean algorithm. Now it holds that $s_i \cdot \prod_{j=1, j\neq i}^{n_{\MRD}} q_j \equiv 1 \bmod q_i$ as $q_i$ and $\prod_{j=1, j\neq i}^{n_{\MRD}} q_j$ are coprime, and $s_i \cdot \prod_{j=1, j\neq i}^{n_{\MRD}} q_j \equiv 0 \bmod q_j$ for $j \neq i$. Thus, the smallest solution corresponds to $z = \sum_{i=1}^{n_{\MRD}} x_i \cdot s_i \cdot \prod_{j=1, j\neq i}^{n_{\MRD}} q_j$ due to the Chinese Remainder theorem \cite{wagon1999mathematica}. Comparing $z$ to the bound $\zeta$ finally yields the answer. Also note that if $n_{\MRD}$ is constant, we can solve the problem by testing all possible vectors $(v_1, \dots, v_{n_{\MRD}})$ with $v_i \in \{x_i, y_i\}$ and then use the Chinese Remainder theorem as explained above. 

\begin{theorem}\label{t:QCtoMRD}
The \QC{} problem is reducible to the \Rest{} problem in polynomial time with the properties that $n_{\MRD} \in O(n_{\QCInd})$, $\max_{i \in \{1, \dots, n_{\MRD}\}}\{q_i, x_i, y_i\}=O(\max_{j \in \{1, \dots, n_{\QCInd}\}}\{b_j^{\beta_j}\}$, and $\zeta \in O(\gamma)$.
\end{theorem}
\begin{proof}
\textit{Transformation:}
Set $q_1 = b_1^{\beta_1}, \dots, q_{n_{\MRD}} = b_{n_{\QCInd}}^{\beta_{\QCInd}}$ and $\zeta = \gamma$ where $\beta_i$ denotes the occurrence of the prime factor $b_i$ in the prime factorization of $\beta$.
Compute $\alpha_i \equiv \alpha \bmod q_i$.
Set $x_i^2 = \alpha_i$ if there exists such an $x_i \in \mathbb{Z}_{q_i}$. Further, compute $y_i = -x_i+q_i$. If there is no such number $x_i$ and thus $y_i$, produce a trivial no-instance. \\

\noindent \textit{Instance size:} The numbers we generate in the reduction equal the prime numbers of the \QC{} problem including their occurrence. Hence, it holds that $\max_{i \in \{1, \dots, n_{\MRD}\}}\{q_i\}$ $=O(\max_{j \in \{1, \dots, n_{\QCInd}\}}\{b_j^{\beta_j}\}$. Due to the modulo, this value also bounds $x_i$ and $y_i$. The upper bound on a solution equals the ones from the instance of the \QC{} problem, \ie,  $\zeta \in O(\gamma)$, and $n_{\MRD} = n_{\QCInd}$ holds.\\

\noindent \textit{Correctness:}
First, let us verify that producing a trivial no-instance is correct if we cannot find some $x_i$. Indeed, this can be traced back to the Chinese Remainder theorem: If and only if there is an $x$ with $x^2 \equiv \alpha \bmod \beta$ and $q_1, \dots, q_{n_{\MRD}}$ (\ie, the equivalences to $b_i^{\beta_i}$) is the prime factorization of $\beta$, then $x^2 \equiv \alpha_i \bmod q_i$, $\alpha_i \in \mathbb{Z}_{q_i}$ for all $i$. In other words, it is has to be dividable by all $b_i^{\beta_i}$ yielding the same remainder $\alpha$ (modulo $b_i^{\beta_i}$). Hence, if there does not exists a square root of $\alpha$ in one of the systems then $x^2 \equiv \alpha \bmod \beta$ has no solution.

But if there exists $x_i$ and $y_i$, these values are in $\mathbb{Z}_{q_i}$ as $x_i \leq \alpha_i < q_i$ per definition of $x_i$ and $\alpha_i$. Further, both values solve the problem $x_i^2, y_i^2 \equiv \alpha \bmod q_i$ as $x_i^2 \equiv \alpha_i \bmod q_i \equiv \alpha_i + \lambda \cdot q_i \bmod q_i \equiv \alpha \bmod q_i$
for some $\lambda \in \mathbb{N}$.
Moreover,

\begin{align*}
& y_i^2 \equiv (-x_i + q_i)^2 \bmod q_i = q_i^2 - 2x_iq_i + x_i^2 \bmod q_i \\
& \equiv x_i^2 \bmod q_i \equiv \alpha \bmod q_i.
\end{align*}
The third equation holds as each summand except the last one is a multiple of $q_i$. The last transformation is true due to the computation above.

Note that for all prime numbers greater than $2$ it holds that $x_i \neq y_i$. This can easily be seen as we already argued that $x_i$ and $y_i$ are in $\mathbb{Z}_{p_i}$. Let us suppose both values are equal, \ie,
\begin{align*}
& x_i^2 = y_i^2 \\
& \Leftrightarrow \alpha_i = (-x_i + q_i)^2 \\
& \Leftrightarrow \alpha_i = q_i^2 - 2q_ix_i + x_i^2 \\
& \Leftrightarrow \alpha_i = q_i^2 - 2q_ix_i + \alpha_i \\
& \Leftrightarrow 2q_ix_i = q_i^2\\
& \Leftrightarrow 2x_i = q_i .
\end{align*}
The factor $q_i$ is a product of some prime number greater than $2$ by the assumption above. Thus, there is no $x_i $ satisfying the formula. \\

\noindent Let us now prove the equivalence of the reduction.\\

$\Rightarrow$ Let the instance of the \QC{} problem be a yes-instance. Then there exists a $z$ satisfying $z^2 \equiv \alpha \bmod \beta$ with $0 < z \leq \gamma$. This solution directly corresponds to a solution of the generated instance of the \Rest{} problem. First, $z \leq \gamma = \zeta$. Secondly, $z$ satisfies all equations as it holds that 
\begin{align*}
 z^2 \equiv \alpha \bmod \beta \equiv \alpha \bmod \prod_{i=1}^{n_{\MRD}} b_i^{\beta_i} \equiv \alpha \bmod b_i^{\beta_i} \text{\, for all \,} i.
\end{align*}
The first equivalence holds as the $b_i^{\beta_i}$s are the prime factorization of $\beta$. The second equivalence is true as we can decompose the solution as follows: $z^2 = \lambda \cdot \prod_{i=1}^{n} b_i^{\beta_i} + \alpha$ for some $\lambda \in \mathbb{N}$. Thus, the first summand is not only divided without remainder by $\prod_{i=1}^{n_{\MRD}} b_i^{\beta_i}$ but also by all primes along with their occurrences alone, leaving only the second summand $\alpha$ as the remainder. Further, since $x_i^2, y_i^2 \equiv \alpha \bmod q_i$ as shown before, it holds that
\begin{align*}
    z^2 \equiv \alpha \bmod b_i^{\beta_i} \equiv \alpha \bmod q_i \equiv x_i^2 \equiv y_i^2  \text{\, for all \,} i. 
\end{align*}
Hence, this satisfies all equations of the generated instance of the \Rest{} problem making it a yes-instance.

$\Leftarrow$ Let the instance of the \Rest{} problem be a yes-instance. Hence, we could verify that there exists a solution to the given equations smaller than $\zeta$. Let this solution be denoted as $z^*$. It holds that $z^* \equiv x_i \bmod q_i$ or $z \equiv y_i \bmod q_i$. Let $v_i$ correspond to the residue that was satisfied, \ie, $v_i = x_i$ or $v_i = y_i$. The solution $z^*$ also solves the \QC{} problem. First, $z^* \leq \zeta = \gamma$. Further, it holds per definition of the numbers that
\begin{align*}
    (z^*)^2 \equiv (v_i)^2 \equiv \alpha \bmod q_i \text{\, for all \,} i. 
\end{align*}
As it satisfies all equations simultaneously and the $b_i$ are pairwise co-prime, it follows from the Chinese Remainder theorem that
\begin{align*}
  &  (z^*)^2 \equiv (v_i)^2 \equiv \alpha \bmod q_i \text{\, for all \,} i \\
   & \equiv (z^*)^2 \equiv \alpha \bmod \prod_{i=1}^{n_{\MRD}} q_i \equiv \alpha \bmod \prod_{i=1}^{n_{\QCInd}} b_i^{\beta_i} \equiv \alpha \bmod \beta
\end{align*}
as the $b_i^{\beta_i}$s are the prime factorization of $\beta$. \\

\noindent \textit{Running time:} Setting the variables accordingly can be done in time polynomial in $n_{\QCInd}$. Further, computing each $x_i, y_i$ can be done in poly-logarithmic time regarding the largest absolute number for each $i \in \{1, \dots, n_{\MRD}\}$~\cite{crandall2006prime}.
\end{proof}

Finally, we reduce the \Rest{} problem to the \TwoStageProb{} problem. Note that the considered \TwoStageProb{} problem is a decision problem. In other words, we only seek to determine whether there exists a feasible solution. We neither optimize a solution vector nor are we interested in the solution vector itself.

\begin{theorem} \label{t:MRDtoTwoStage}
The \Rest{} problem is reducible to the  \TwoStageProb{} problem in polynomial time with the properties that $n \in O(n_{\MRD})$, $r,s,t, ||c||_{\infty}, ||b||_{\infty}, ||\ell||_{\infty} \in O(1)$, $||u||_{\infty} \in O(\zeta)$, and $\Delta \in O(\max_i\{q_i\})$.
\end{theorem}
\begin{proof}
\textit{Transformation:} Having the instance for the \Rest{} problem at hand we construct our ILP as follows with $n = n_{\MRD}$:

\begin{equation*}
\mathcal A \cdot x =
\begin{pmatrix}
-1       & q_1  & x_1  & y_1  &0&\dots&0& 0     & \dots & 0 \\
0       & 0     & 1     & 1     &0&\dots&0& 0     & \dots & 0 \\
\vdots  &\vdots &\ddots &\ddots &\ddots&\ddots&\ddots&\ddots &\ddots &\ddots \\
-1       & 0     & \dots & 0     &0&\dots&0& q_n  & x_n  & y_n \\
0       & 0     & \dots & 0     &0&\dots&0& 0     & 1     & 1 \\
\end{pmatrix} \cdot x = b =
\begin{pmatrix}
0 \\
1 \\
\vdots \\
0 \\
1 \\
\end{pmatrix}.
\end{equation*}
All variables get a lower bound of $0$ and an upper bound of $\zeta$. We can set the objective function arbitrarily as we are just searching for a feasible solution, hence we set it to $c = (0, 0, \dots, 0)^\top$. \\

\noindent \textit{Instance size:} Due to our construction, it holds that $t = 2, r = 1, s = 3$. The number $n$ of repeated blocks equals the number $n_{\MRD}$ of equations in the instance of the \Rest{} problem. The largest entry $\Delta$ can be bounded by $\max_i\{q_i\}$. The lower and upper bounds are at most $||u||_{\infty} = O(\zeta)$, $||\ell||_{\infty} = O(1)$. The objective function $c$ is set to zero and is thus of constant size. The largest value in the right-hand side is $||b||_\infty = 1$.\\

\noindent \textit{Correctness:}
$\Rightarrow$ Let the given instance o the \Rest{} problem be a yes-instance. Thus, there exists a solution $z^* < \zeta$ satisfying all equations. As before, let $v_i$ correspond to the remainder that was satisfied in each equation~$i$, \ie, $v_i = x_i$ or $v_i = y_i$. A solution to our integer linear program now looks as follows: Set the first variable to $z^*$. Let the columns corresponding to $x_i$ and $y_i$ be set as follows for each $i$: If $v_i = x_i$ then set this variable occurrence in the solution vector to 1. Set the occurrence to the corresponding variable of $y_i$ to zero. Otherwise, set the variables the other way round. Finally, the variable corresponding to the columns of the $q_i$ are computed as $(z^* - v_i)/q_i$. It is easy to see that this solution is feasible and satisfies the bounds on the variable sizes. \\  

$\Leftarrow$ Let the given instance of the \TwoStageProb{} problem be a yes-instance. By definition of the constraint matrix we have for every $1 \leq i \leq n$ that there exists a multiple $\lambda_i \geq 0$ such that $z = x_i + \lambda_i q_i$ or $z = y_i + \lambda_i q_i$. Hence $z \equiv x_i \mod q_i$ or $z \equiv y_i \mod q_i$ for every $1 \leq i \leq n$. Further, $z \leq u$. Thus, the solution $z$ is a solution of the \Rest{} problem. \\

\noindent \textit{Running time:} Mapping the variables and computing the values for the $q_i$s can all be done in polynomial time regarding the largest occurring number and $n$. 
\end{proof}

\section{Runtime Bounds for $2$-Stage Stochastic ILPs under ETH} \label{sec:ETH}
This sections presents the proof that the double exponential running time in the current state-of-the-art algorithms is nearly tight assuming the Exponential Time Hypothesis (ETH). To do so, we make use of the reductions above showing that we can transform an instance of the $3$-\textsc{SAT} problem to an instance of the \TwoStageProb{} problem.

\begin{corollary} \label{t:ETHParamN}
The \TwoStageProb{} problem cannot be solved in time less than $2^{\delta \sqrt{n}}$ for some $\delta > 0$ assuming ETH.
\end{corollary}
\begin{proof}
Suppose the opposite. That is, there is an algorithm solving the \TwoStageProb{} problem in time less than $2^{\delta \sqrt{n}}$.
Let an instance of the $3$-\textsc{SAT} problem with $n_\SatInd$ variables and $m_\SatInd$ clauses be given. Due to the Sparsification lemma, we may assume that $m_\SatInd \in O(n_\SatInd)$ \cite{DBLP:journals/jcss/ImpagliazzoPZ01}. The Sparsification lemma states that any $3$-SAT formula can be replaced by subexponentially many $3$-SAT formulas, each with a linear number of clauses with respect to the number of variables. The original formula is satisfiable if at least one of the new formulas is. This yields that if we cannot decide a $3$-SAT problem in subexponential time, we can also not do so for a $3$-SAT problem where $m_\SatInd \in O(n_\SatInd)$.

We can reduce such an instance to an instance of the \QC{} problem in polynomial time regarding $n_\SatInd$ such that $n_{\QCInd} \in O(n_\SatInd^2)$, $max_{i}\{b_i\} \in O(n_\SatInd^2\log(n_\SatInd))$, $\alpha, \beta, \gamma = 2^{O(n_\SatInd^2 \log(n_\SatInd))}$, see Theorems~\ref{t:3SattoQCRed} and \ref{t:3SattoQCInstanceSize}.

Next, we reduce this instance to an instance of the \Rest{} problem. Using Theorem \ref{t:QCtoMRD}, this yields the parameter sizes $n_{\MRD} \in O(n_\SatInd^2)$, $\max_{i \in \{1, \dots, n_{\MRD}\}}\{q_i, x_i, y_i\}=O(n_\SatInd^2\log(n_\SatInd))$, and finally $\zeta \in$ $ 2^{O(n_\SatInd^2 \log(n_\SatInd))}$. Note that all prime numbers greater than $2$ appear at most once in the prime factorization of $\beta$ and $2$ appears $4$ times. Thus, the largest $q_i$, which corresponds to $\max_i\{b_i^{\beta_i}\}$ equals the largest prime number in the \QC{} problem: The largest prime number is at least the $(\nu^2+2\nu+2m'+13) \geq 13$th prime number by a rough estimation. The $13$th prime number is $41$ and thus larger than $2^4 = 16$. 

Finally, we reduce that instance to an instance of the \TwoStageProb{} problem with parameters $r,s,t, ||c||_{\infty}$, $ ||b||_{\infty}, ||\ell||_{\infty} \in O(1)$, $||u||_{\infty} \in  2^{O(n_\SatInd^2 \log(n_\SatInd))}$, $n \in O(n_\SatInd^2)$, and $\Delta \in O(n_\SatInd^2 $ $\log(n_\SatInd))$, see Theorem \ref{t:MRDtoTwoStage}.

Hence, if there is an algorithm solving the \TwoStageProb{} problem in time less than $2^{\delta \sqrt{n}}$ this would result in the $3$-\textsc{SAT} problem to be solved in time less than $2^{\delta \sqrt{n}} = 2^{\delta \sqrt{C_1 n_\SatInd^2}} = 2^{\delta(C_2 n_\SatInd))}$ for some constants $C_1$, $C_2$. Setting $\delta_\SatInd \leq \delta/C_2$, this would violate the ETH.
\end{proof}

To prove our main result, we still have to reduce the size of the coefficients in the constraint matrix. To do so, we encode large coefficients into submatrices. This reduces the size of the entries greatly while just extending the matrix dimensions slightly. A similar approach was used for example in~\cite{DBLP:conf/ipco/Klein20} to prove a lower bound for the size of inclusion minimal kern-elements of \TwoStage{} ILPs or in \cite{DBLP:conf/stacs/KnopPW19} to decrease the value of $\Delta$ in the matrices.

\begin{theorem}
The \TwoStageProb{} problem cannot be solved in time less than $2^{2^{\delta (s+ t)}} |I|^{O(1)}$ for some constant $\delta > 0$, even if $r=1$, $\Delta, ||b||_{\infty}, ||c||_{\infty}, ||b||_{\infty} \in O(1)$, assuming ETH. Here $|I|$ denotes the encoding length of the total input. 
\end{theorem}
\begin{proof}
First, we show that we can alter the resulting integer linear program such that we reduce the size of $\Delta$ to $O(1)$. We do so by encoding large coefficients with base $2$, which comes at the cost of enlarged dimensions of the constraint matrix. Let $\enc(x)$ be the encoding of a number $x$ with base $2$. Further, let $\enc_i(x)$ be the $i$th number of $\enc(x)$. Finally, $\enc_0(x)$ denotes the last significant number of the encoding. Hence, the encoding of a number $x$ is $\enc(x) = \enc_0(x) \enc_1(x) \dots \enc_{\lfloor\log(\Delta)\rfloor}(x)$ and $x$ can be reconstructed by $x = \sum\nolimits^{\lfloor\log(\Delta)\rfloor}_{i = 0} \enc_i(x) \cdot 2^i$.

Let a matrix $E$ be defined as,
\begin{equation*}
E = 
\begin{pmatrix}
2   & -1         & 0    &  \dots & 0  \\
0               & 2 & -1   & 0 \dots  &   0 \\
\vdots & \ddots & \ddots  & \ddots\\
0 & \dots & 0 & 2 & -1 

\end{pmatrix}.
\end{equation*}

We re-write the constraint matrix as follows: For each coefficient $a > 1$, we insert its encoding $\enc(a)$ and beneath we put the matrix $E$. Furthermore, we have to fix the dimensions for the first row in the constraint matrix, the columns without great coefficients and the right-hand side $b$ by filling the matrix at the corresponding positions with zeros. The altered integer linear program $\mathcal{A}\cdot x = b$ is displayed in Figure~\ref{fig:AlteredMatrix}.
\begin{figure}[h!]
\begin{equation*}
\begin{pmatrix}
-1       & \enc(q_1) & \enc(x_1)  & \enc(y_1)  &0&\dots&0& 0     & \dots & 0 \\

0 &  E & 0\dots0& 0\dots0&0&\dots&0& 0     & \dots & 0\\

\vdots & 0\dots0& E & 0\dots0 &0&\dots&0& 0     & \dots & 0 \\

0  &0\dots0 & 0\dots0& E &0&\dots&0& 0     & \dots & 0\\

0       & 0\dots0     & 1\,0 \dots 0     & 1\,0 \dots 0     &0&\dots&0& 0     & \dots & 0 \\

\vdots  &\vdots &\ddots &\ddots &\ddots&\ddots&\ddots&\ddots &\ddots &\ddots \\

-1       & 0     & \dots & 0     &0&\dots&0& \enc(q_n)  & \enc(x_n)  & \enc(y_n) \\

0 &0&\dots&0& 0     & \dots & 0 & E &0\dots0 &0\dots0\\

\vdots  &0&\dots&0& 0     & \dots & 0 &0\dots0 & E &0\dots0 \\

0  &0&\dots&0& 0     & \dots & 0 &0\dots0 &0\dots0 & E\\

0       & 0     & \dots & 0     &0&\dots&0& 0\,\dots0     & 1\,0 \dots 0       & 1\,0 \dots 0
\end{pmatrix} \cdot x =
\begin{pmatrix}
0 \\
\vdots \\
 \\
0 \\
1 \\
\vdots\\
0 \\
\vdots \\
 \\
0 \\
1 \\
\end{pmatrix}
\end{equation*}.
\caption{The displayed ILP is the altered ILP after encoding large entries with basis $2$.} \label{fig:AlteredMatrix}
\end{figure}
Note that the ones beneath the sub-matrices $\enc(x_i)$ and $\enc(y_i)$ correspond to $\enc_0(x_i)$ and $\enc_0(y_i)$.
The independent blocks consisting of $\enc(a)$ and the matrix $E$ beneath correctly encodes the number $a>1$, \ie, it preserves the solution space: Let $x_a$ be the number in the solution corresponding to the column with entry $a$ of the original instance. The solution for the altered column (\ie, the sub-matrix) is $(x_a \cdot 2^0, x_a \cdot 2^1, \dots, x_a \cdot 2^{\lfloor\log(\Delta)\rfloor})$. The additional factor of $2$ for each subsequent entry is due to the diagonal of $E$. It is easy to see that $a \cdot x_a = \sum\nolimits_{i = 0}^{\lfloor\log(\Delta)\rfloor} \enc_i(a) \cdot x_a \cdot 2^i$ as we can extract $x_a$ on the right-hand side and solely the encoding of $a$ remains. Thus, the solutions of the original matrix and the altered one directly transfer to each other. Hence, the solution space is preserved.

Regarding the dimensions, each coefficient $a > 1$ is replaced by a $(O(\log(\Delta)) \times O(\log(\Delta)))$ matrix. Thus, the dimension expands to $t' = t \cdot O(\log(\Delta)) = O(\log(\Delta))$, $s' = s \cdot O(\log(\Delta)) = O(\log(\Delta))$, while $r$ and $n$ stay the same. Further, we have to adjust the bounds. The lower bound for all new variables is also zero. For the upper bounds we allow an additional factor of $2^i$ for the $i$th value of the encoding. Thus, $||u'||_{\infty} = 2^{{\lfloor\log(\Delta)\rfloor}} ||u||_{\infty}$. Further, we get that the largest coefficient is bounded by $\Delta' = O(1)$. The right-hand side $b$ enlarges to a vector $b'$ with $O(n\log(\Delta))$ entries.

Now suppose there is an algorithm solving the \TwoStageProb{} problem in time less than $2^{2^{\delta (s+t)}} |I|^{O(1)}$.
The proof of Theorem \ref{t:ETHParamN} shows that we can transform an instance of the $3$-\textsc{SAT} problem with $n_\SatInd$ variables and $m_\SatInd$ clauses to an \TwoStage{} ILP with parameters $r,s,t, ||c||_{\infty}, ||b||_{\infty}, ||\ell||_{\infty} \in O(1)$, $||u||_{\infty} \in 2^{O(n_\SatInd^2 \log(n_\SatInd))}$, $n \in O(n_\SatInd^2)$, and $\Delta \in O(n_\SatInd^2\log(n_\SatInd))$. Further, we explained above that we can transform this ILP to an equivalent one where 
\begin{align*}
& t' = O(\log(\Delta)) = O(\log(n_\SatInd^2 \log(n_\SatInd))) = O(\log(n_\SatInd)), \\
& s' = O(\log(\Delta)) = O(\log(n_\SatInd^2 \log(n_\SatInd))) = O(\log(n_\SatInd)), \\
& \Delta' = O(1),\\
& b' \in \mathbb{Z}^{O(n_\SatInd^2\log(n_\SatInd))},\\
&||u'||_{\infty} = 2^{{\lfloor\log(\Delta)\rfloor}} ||u||_{\infty} =  2^{{\lfloor\log(n_\SatInd^2\log(n_\SatInd))\rfloor}}  2^{O(n_\SatInd^2 \log(n_\SatInd))} =   2^{O(n_\SatInd^2 \log(n_\SatInd))},
\end{align*}
while $r$, and $n$ stay the same.
The encoding length $|I|$ is then given by 
\begin{align*}
|I| = (nt'(r+ns'))\log(\Delta')+(r+ns')\log(||\ell||_{\infty})+ \\
(r+ns')\log(||u'||_{\infty}) +nt'\log(||b'||_{\infty}) +(r+ns')\log(||c||_{\infty}) \\
= 2^{O(n_\SatInd^2)}.
\end{align*}

 Hence, if there is an algorithm solving the \TwoStageProb{} problem in time less than $2^{2^{\delta (s+t)}} |I|^{O(1)}$ this would result in the $3$-\textsc{SAT} problem to be solved in time less than 
\begin{align*}
2^{2^{\delta (s+t)}} |I|^{O(1)}= 2^{2^{\delta ( C_1\log(n_\SatInd) + C_2\log(n_\SatInd)) }}  2^{n_\SatInd^{O(1)}} = 2^{2^{\delta C_3\log(n_\SatInd) }} 2^{n_\SatInd^{O(1)}} \\
= 2^{n_\SatInd^{\delta \cdot C_3}} 2^{n_\SatInd^{O(1)}} = 2^{n_\SatInd^{\delta \cdot C_4}}
\end{align*}
 for some constants $C_1, C_2, C_3, C_4$. Setting $\delta = \delta'/C_4$ we get $2^{n_\SatInd^{\delta C_4}} = 2^{n_\SatInd^{\delta'}}$. As it holds for sufficient large $x$ and $\epsilon < 1$ that $x^\epsilon < \epsilon x$ it follows that $2^{n_\SatInd^{\delta'}} < 2^{\delta'n_\SatInd}$. This violates the ETH. Note that this result even holds if $r=1$, $\Delta, ||c||_{\infty}, ||b||_{\infty},||\ell||_{\infty} \in O(1)$ as constructed by our reductions. 
\end{proof}

\bibliography{ref}
\end{document}